\documentclass[journal/twoside,web]{ieeecolor}
\usepackage{lcsys}
\usepackage{cite}
\usepackage{amsmath,amssymb,amsfonts}
\usepackage{graphicx}
\usepackage{textcomp}
\usepackage{amssymb}
\usepackage{amsmath}
\usepackage{graphicx}
\usepackage{xcolor}
\usepackage{makecell}
\usepackage{comment}
\usepackage{todonotes}

\usepackage{cite}
\usepackage{graphicx}
\usepackage{subcaption}

\usepackage{xcolor}
\definecolor{light-blue}{rgb}{0.3,0.5,0.8}
\usepackage[colorlinks=true, linkcolor=cyan, citecolor=cyan, filecolor=magenta, urlcolor=cyan]{hyperref}

\def\BibTeX{{\rm B\kern-.05em{\sc i\kern-.025em b}\kern-.08em
    T\kern-.1667em\lower.7ex\hbox{E}\kern-.125emX}}
\title{Performance-Sensitive Potential Functions for Efficient Flow of 
Connected and Automated Vehicles}
\author{Filippos N. Tzortzoglou, \IEEEmembership{Student member, IEEE,} Dionysios Theodosis, Aditya Dave, \IEEEmembership{Member, IEEE}
and Andreas A. Malikopoulos,
\IEEEmembership{Senior Member, IEEE}
\thanks{This work was supported by NSF under Grants CNS-2149520 and CMMI-2219761.} \thanks{Filippos N. Tzortzoglou, Aditya Dave, and Andreas A. Malikopoulos are with the Department of Civil and Environmental Engineering, Cornell University, Ithaca, NY 14853 USA. Dionysios Theodosis is with the Decision and Simulation Laboratory, Technical University of Crete, Chania, Greece  (emails: \tt\small{ft254@cornell.edu; dtheodosis@tuc.gr; a.dave@cornell.edu; amaliko@cornell.edu)}}}

\usepackage{amsthm}

\newtheorem{lemma}{Lemma} 
\newtheorem{theorem}{Theorem}

\newtheorem{remark}{Remark}

\begin{document}
\maketitle
  \thispagestyle{empty}
  \pagestyle{empty}
\begin{abstract}
Connected and automated vehicles (CAVs) provide the most intriguing opportunity for enabling users to  monitor transportation network conditions and make better decisions for improving safety and transportation efficiency.
In this paper, we address the problem of effectively coordinating CAVs on lane-based roadways.
Our approach utilizes potential functions to generate repulsive forces between CAVs that ensure collision avoidance.
However, such potential functions can lead to unrealistic acceleration profiles and large inter-vehicle distances.
The primary contribution of this work is the introduction of performance-sensitive potential functions to address these challenges. 
In our approach, the parameters of a potential function are determined through an optimization problem aiming to reduce both acceleration and inter-vehicle distances. To circumvent the computational implications due to the complexity of the resulting optimization problem that prevents the derivation of a real-time solution, we train a neural network model to learn the mapping of initial conditions to optimal parameters derived offline. Then, we prove sufficient criteria for the sampled-data model to ensure that the neural network output does not activate any of the state and safety constraints. Finally, we provide simulation results to demonstrate the effectiveness of the proposed approach.
\end{abstract}


\section{INTRODUCTION}

Connected and automated vehicles (CAVs) have the potential to significantly enhance the performance of transportation networks, including improvements in safety, comfort, energy efficiency, and reduction of congestion \cite{othman2019ecological, rios2016survey, Malikopoulos2020}.  
A fundamental approach towards these factors is to develop novel control strategies for CAVs that can optimize their performance and manage the impacts of vehicle interactions on system behavior in a variety of traffic situations \cite{malikopoulos2018decentralized, xiao2021bridging}.

Research has shown significant benefits when using adaptive cruise control (ACC) \cite{xiao2010comprehensive} to automatically regulate the car's speed, keeping a safe distance from the vehicle ahead. Originally, the primary focus of ACC system design was safety, specifically to prevent collisions \cite{ioannou1993autonomous}, and improve traffic flow \cite{van2006impact}. However, recent developments have added fuel and energy efficiency as another important performance criterion \cite{bae2019design}. The advent of vehicle-to-everything connectivity has led to several interesting approaches to cooperative ACC (CACC) \cite{zhang2020cooperative} and ecological CACC systems \cite{zhai2018ecological}. Ecological CACC has demonstrated significant potential for energy savings in different driving conditions, e.g., highways, arterial roads with or without traffic signals, and at highway merging \cite{vahidi2018energy}.  Some of these studies have implemented and tested their controllers in real-world driving scenarios \cite{mahbub2020sae-1}.

In prior research efforts  \cite{karafyllis2022stability, theodosis2022sampled, tzortzoglou2023approach, karafyllis2022lyapunov}, it has been shown that the behavior of CAVs can be effectively coordinated on single-lane and lane-free roads using control Lyapunov functions. Since this results in decentralized feedback laws, such an approach facilitates each vehicle to determine its control input based on its speed, relative speeds, distances from adjacent vehicles, and the boundary of the road. 
The Lyapunov function proposed in \cite{karafyllis2022stability} relies upon a term that induces a repulsive potential between vehicles to avoid collisions.
This term constitutes a potential function that can negatively impact the traffic flow by (1) increasing inter-vehicle distances and (2) negatively influencing accelerations to decrease the passenger comfort and the fuel efficiency of all vehicles. 

This paper addresses these issues by improving the potential function introduced in \cite{karafyllis2022stability} to develop better vehicle behavior for both acceleration and inter-vehicle distances.
To this end, we focus on scenarios with vehicles operating on single-lane roadways. 
First, we introduce a new family of parameterized potential functions whose shape can be changed by adjusting the parameter values. These functions can induce the desired behavior in vehicles while allowing for greater precision in how that behavior is achieved. 
Then, our goal is to determine the optimal parameters that achieve our objectives.

We formulate an optimization problem with the parameters as the decision variables and an objective comprising of a combination of the acceleration and the inter-vehicle distances.
We solve this optimization problem offline using numerical methods, and then we train a neural network to provide the solutions during real-time implementation. This network is effectively trained on an offline dataset of our problem's solution for various initial conditions. 
In the existing literature, we found no approaches for controlling CAV acceleration and inter-vehicle distances using a similar approach within a Lyapunov framework.
Finally, since we use sampled data to emulate our model, we also prove sufficient conditions to show that the selected sampling ensures both collision avoidance and positive speeds. 
The main contributions of this paper are
(1) the introduction of performance-sensitive potential functions to a) reduce vehicle acceleration and improve efficiency and comfort, and b) moderate inter-vehicle distances and improve traffic flow; and (2) the establishment of conditions for the sampled-data model that guarantee no activation of the state and safety constraints. We validate our analysis using MATLAB Automated Driving Toolbox. 

The structure of the rest of the paper is as follows. In Section II, we present the modeling framework, and in Section III, we present the new potential function. In Section IV, we address the solution approach, and in Section V, we provide simulation results. Finally, we draw concluding remarks in Section VI.


\section{ Modeling Framework}

Consider a single-lane road and a coordinator, which can be a group of loop detectors or comparable sensory devices that can access the state of the road as shown in Fig. \ref{fig:Lane}. Consider $n\in\mathbb{N}$ vehicles operating on the single-lane road. The dynamics of the vehicles are described by the following ordinary differential equations (ODEs):

\begin{equation} \label{eq:initial model}
    \begin{aligned}
      \dot{x}_i =  v_i \\
      \dot{v}_i =  F_i
          \end{aligned}, \quad i=1,\ldots,n,
  \end{equation}
 where   $x_i \in \mathbb{R}$ is the position, $v_i \in \mathbb{R}$ is the speed, and $F_i \in \mathbb{R}$ is the feedback law (acceleration) given by \cite{karafyllis2022stability}: 
  \begin{equation} \label{eq:initial model1}
    \begin{aligned}
      &{F}_1 = -k_1(s_2)(v_1 - v^*) - V'(s_2), \\
      &{F}_i= -k_i(s_i, s_{i+1})(v_i - v^*) + V'(s_i) - V'(s_{i+1}), \\
      &{F}_n = -k_n(s_n)(v_n - v^*) + V'(s_n). \\
          \end{aligned}
  \end{equation}
Here, $s_i:=x_{i-1}-x_i$, $i=2,\ldots,n$ is the inter-vehicle distance, and 
   \begin{equation} \label{eq:initial model2}
    \begin{aligned}   
      &k_1(s_2) = \mu + g(-V'(s_2)), \\
      &k_i(s_i, s_{i+1}) = \mu + g(V'(s_i) - V'(s_{i+1})), \\
      &k_n(s_n) = \mu + g(V'(s_n)),
    \end{aligned}
  \end{equation}
with gain $\mu>0$; the function $g: \mathbb{R} \to \mathbb{R}$ is given by 
\begin{equation} \label{eq:g_function}
g(x)=\frac{v_{\max}f(x)}{v^{*}(v_{\max}-v^{*})}-\frac{x}{v^*},\,\,x\in\mathbb{R}.
\end{equation}
The terms $k_i$ in \eqref{eq:initial model1} and \eqref{eq:initial model2} are state-dependent gains that guarantee the speed of each vehicle remains positive and less than the speed limit $v_{max}$. The constant $v^*\in(0,v_{\max})$ in \eqref{eq:initial model1} denotes the desired cruising speed for each vehicle.  The function $V \in C^2$ in \eqref{eq:initial model2} is a potential function that exerts repulsive force for collision avoidance. As the distance $s$ between two vehicles decreases, the function yields higher values (of repulsion between the vehicles) to prevent the vehicles from colliding. On the other hand, when the distance is greater than a specific value, there is no repulsion. In particular, the function $V$ satisfies the following conditions:

\begin{equation}\label{eq:V_properties}
    \begin{aligned}
\lim_{s \to L^+} V(s) & = +\infty, \\
V(s)& =0, \quad \forall\; s\geq\lambda,        
    \end{aligned}
\end{equation}

\noindent where $L>0$ is the minimum allowable inter-vehicle distance and $\lambda>L$ is the distance at which vehicles will no longer exert repulsive forces on each other. In \eqref{eq:g_function}, the function $f(x)$  is non-decreasing and satisfies $\max(x,0)\leq f(x),\; \forall\,\,x \in \mathbb{R}.$ Here, we select the function $f$ as in \cite{karafyllis2022stability}: 
\begin{equation}
  f(x) =\frac{1}{2\epsilon} \begin{cases}
  0, & \text{if } x \leq -\epsilon, \\
  (x+\epsilon)^2, & \text{if } -\epsilon < x < 0, \\
  \epsilon^2 +2 \epsilon x & \text{if } x \geq 0,
\end{cases}
\end{equation}
with parameter $\epsilon>0$.
The state space of the closed-loop system \eqref{eq:initial model}-\eqref{eq:initial model1} after applying the feedback law is given by  \
\begin{equation} \label{state_space} 
\begin{aligned}\Omega =&\left\{\, (s_{2} ,...,s_{n} ,v_{1} ,...,v_{n} )\in \mathbb{R} ^{2n-1} \, :\, \mathop{\min }\limits_{i=2,...,n} (s_{i} )>L\, ,\right.\\&\left. \mathop{\max }\limits_{i=1,...,n} (v_{i} )\le v_{\max }, 
\mathop{\min }\limits_{i=1,...,n} (v_{i} )\ge 0\, \right\}. \nonumber \end{aligned}
\end{equation} 
The existence of a desirable solution was established in \cite{karafyllis2022stability}.

\newlength{\mycitelength}
\settowidth{\mycitelength}{\cite{karafyllis2022stability}}

\begin{theorem}[\hspace{\mycitelength}\llap{\cite{karafyllis2022stability}}\hspace{3pt}]\label{thm1}
For every initial condition $(s_2(t_0),\ldots,s_n(t_0),v_1(t_0),\ldots,v_n(t_0))\in\Omega$, the solution $(s_2(t),\ldots,s_n(t),v_1(t),\ldots,v_n(t))\in\Omega$ is defined for all $t\ge t_0$ and satisfies $\lim_{t\to+\infty}(v_i(t))=v^*$ for all $i=1,\ldots,n$ and $\lim_{t\to +\infty}(V'(s_i(t))=0$ for all $i=2,\ldots,n $, where $t$ represents the time and $t_0$ denotes the initial time.
\end{theorem}

\begin{figure}[t!]
\centering
\includegraphics[width=0.3\textwidth]{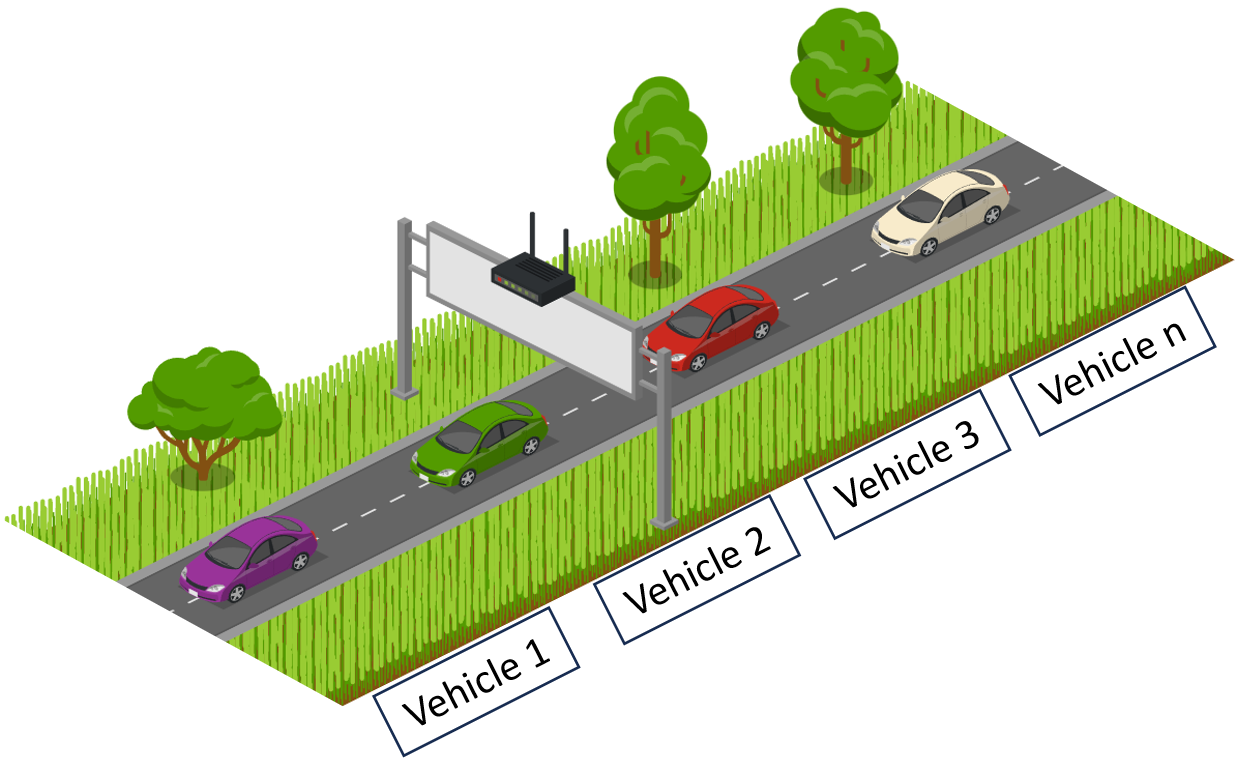}
\caption{Single lane road with four vehicles.}
\label{fig:Lane}
\end{figure}
Since $v^*\in(0,v_{\max})$, we can use conditions \eqref{eq:initial model} and \eqref{eq:V_properties} with $\dot{s}_i=v_{i-1}-v_i$ in place of $\dot{x}_i=v_i$ to establish that the set of equilibrium points is 
\begin{equation} \label{set_S} 
\begin{aligned}
S=&\left\{\, (s_{2} ,...,s_{n} ,v_{1} ,...,v_{n} )\in \mathbb{R} ^{2n-1} \, : v_{i} =v^{*} ,i=1,...,n\, \right.\\
&\left. V'(s_{i})=0, i=2,\ldots,n \right\}\subset \Omega. 
\end{aligned}
\end{equation}



\begin{figure*}[h!]
    \centering
    \begin{subfigure}[b]{0.315\textwidth}
        \includegraphics[width=\textwidth]{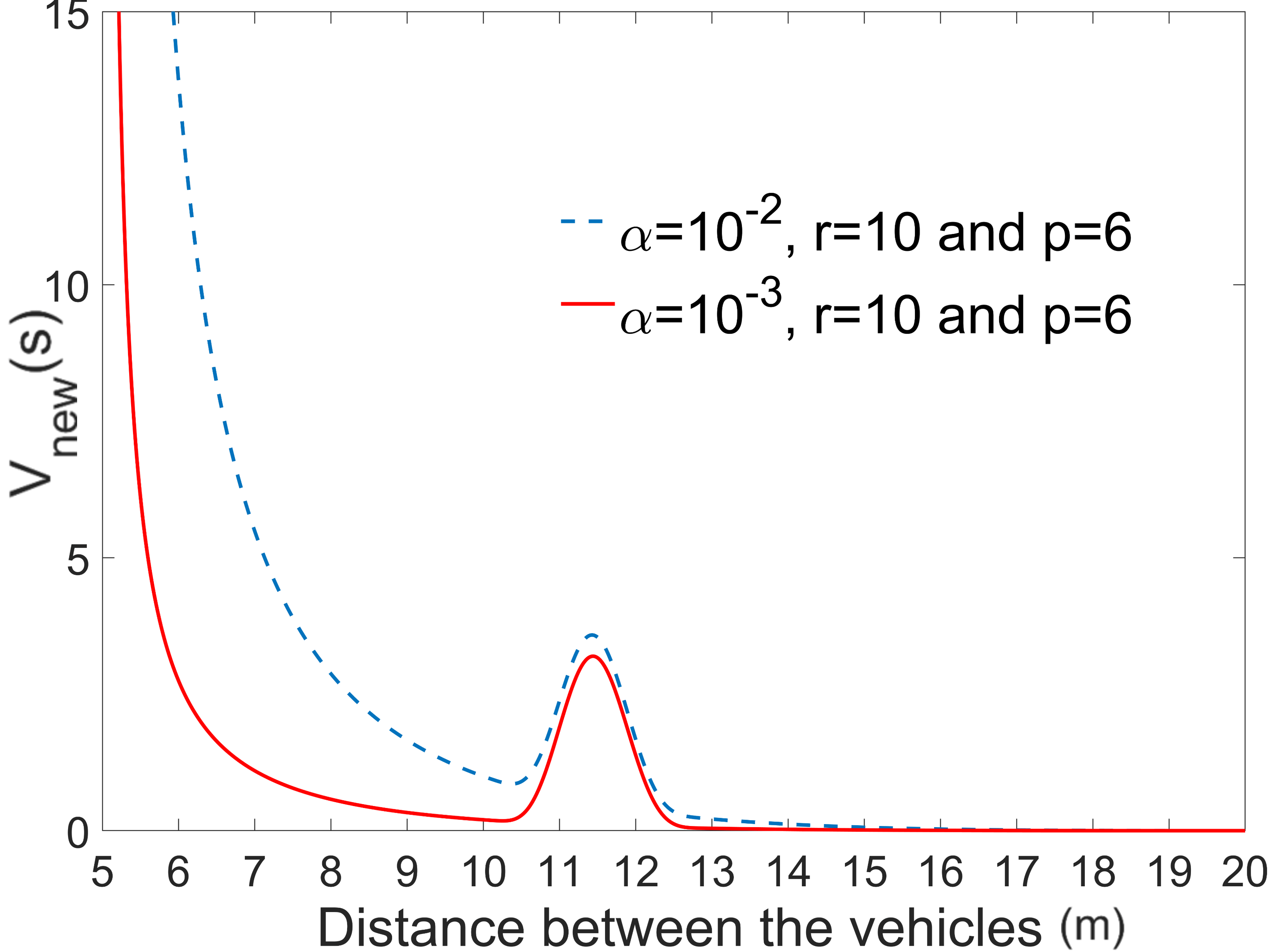}
        \caption{Influence of the parameter $\alpha$.}
        \label{fig:Potentials_a}
    \end{subfigure}
    \hfill
    \begin{subfigure}[b]{0.32\textwidth}
        \includegraphics[width=\textwidth]{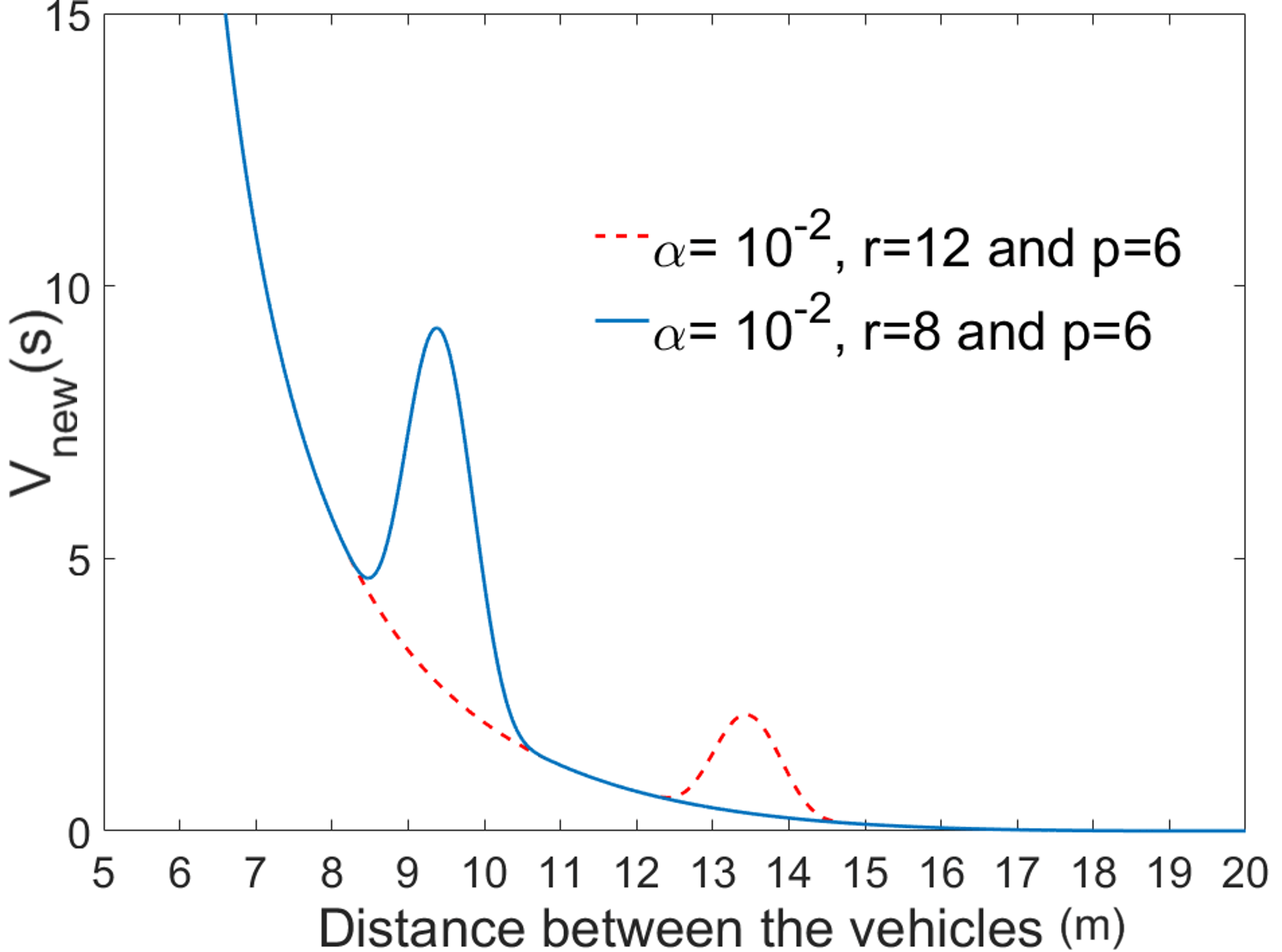}
        \caption{Influence of the parameter $r$.}
        \label{fig:Potentials_r}
    \end{subfigure}
    \hfill
    \begin{subfigure}[b]{0.32\textwidth}
        \includegraphics[width=\textwidth]{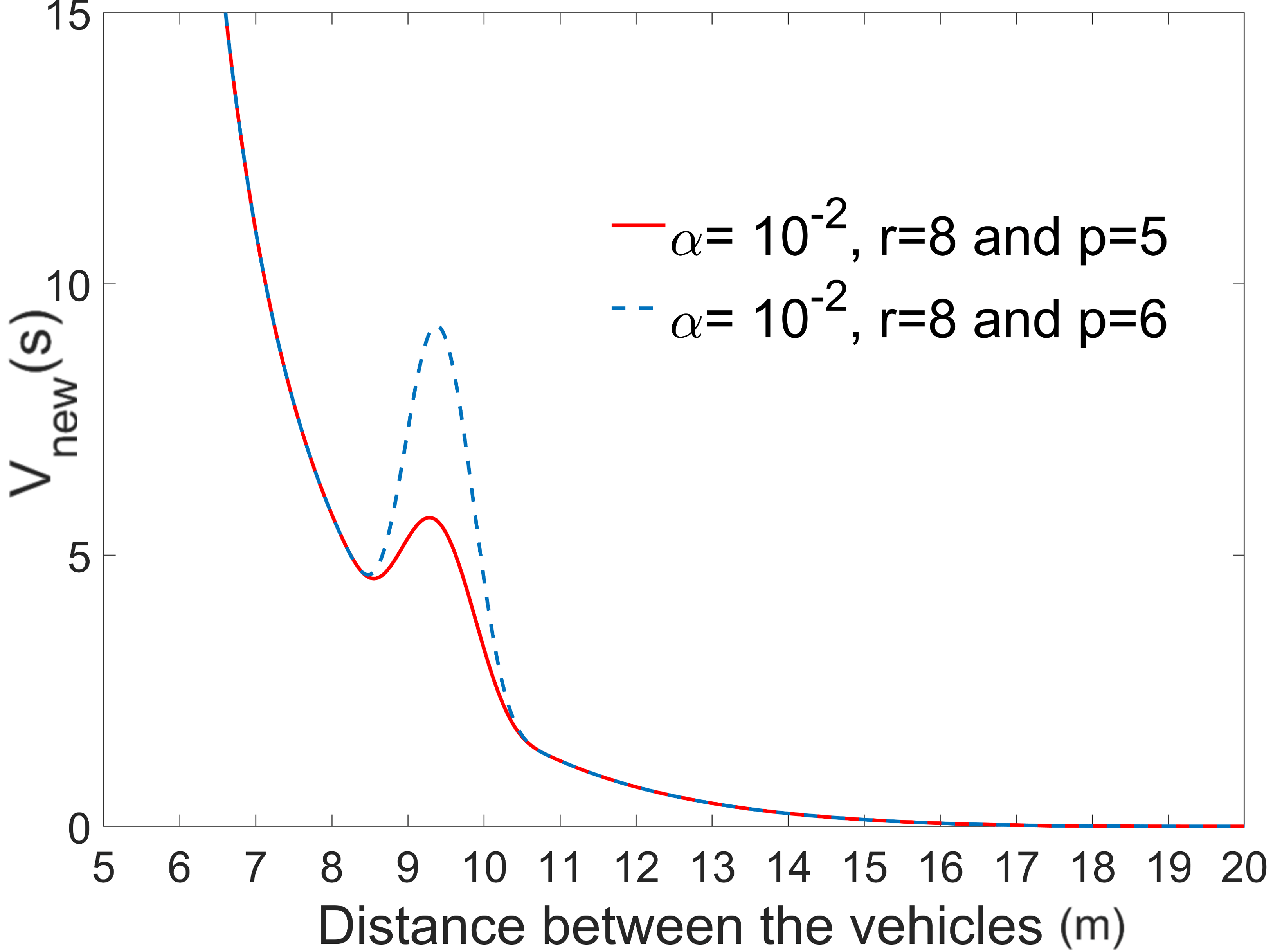}
        \caption{Influence of the parameter $p$.}
        \label{fig:Potentials_p}
    \end{subfigure}
    \caption{The potential $V_{new}$ under different combinations of parameters $\alpha$, $r$ and $p$.}
    \label{fig:Potentials_parameters}
\end{figure*}

\begin{figure*}
    \centering
    \begin{subfigure}[b]{0.32\textwidth}
        \includegraphics[width=\textwidth]{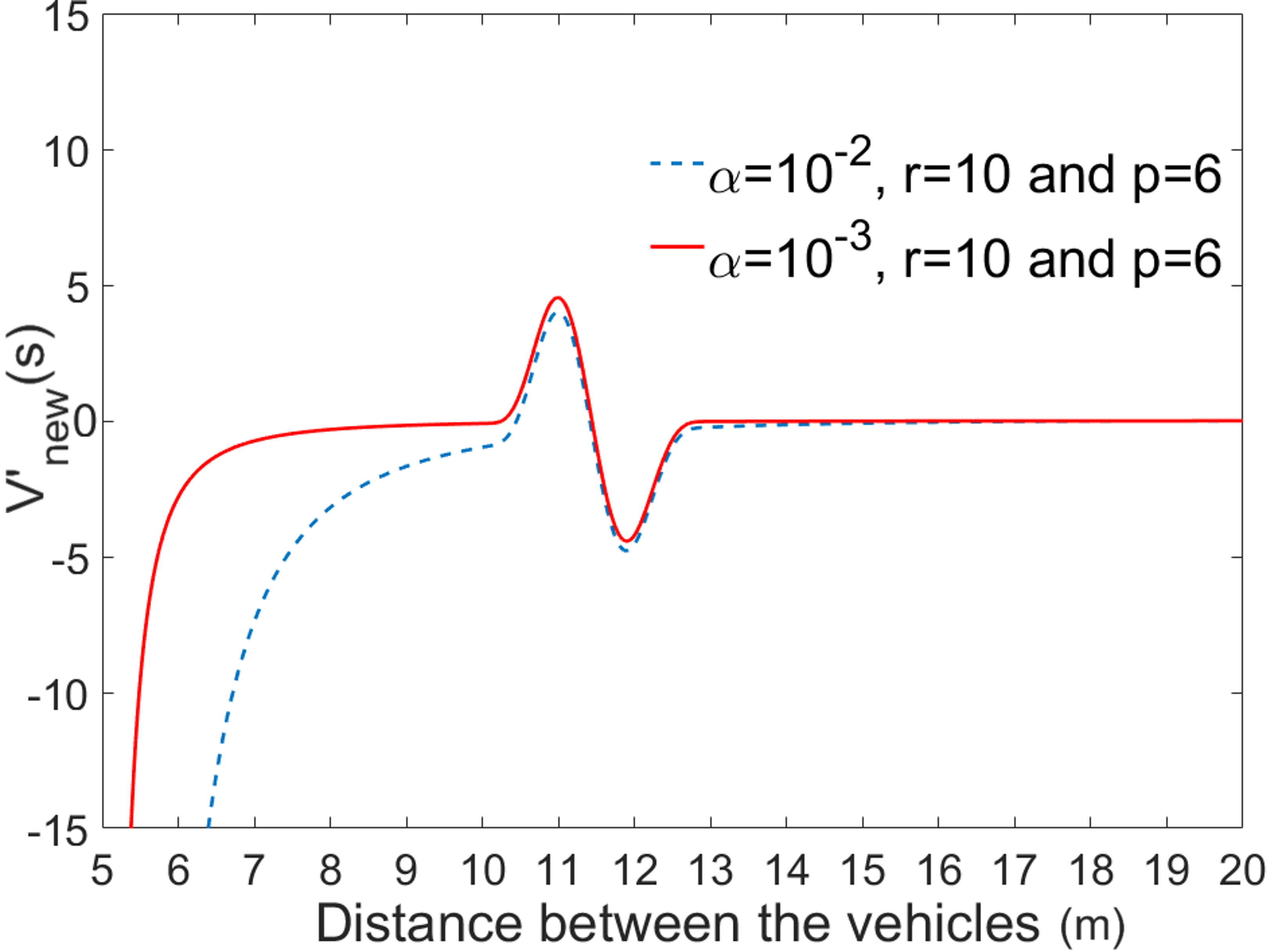}
        \caption{Influence of the parameter $\alpha$.}
        \label{fig:Potentials_a_derivative}
    \end{subfigure}
    \hfill
    \begin{subfigure}[b]{0.32\textwidth}
        \includegraphics[width=\textwidth]{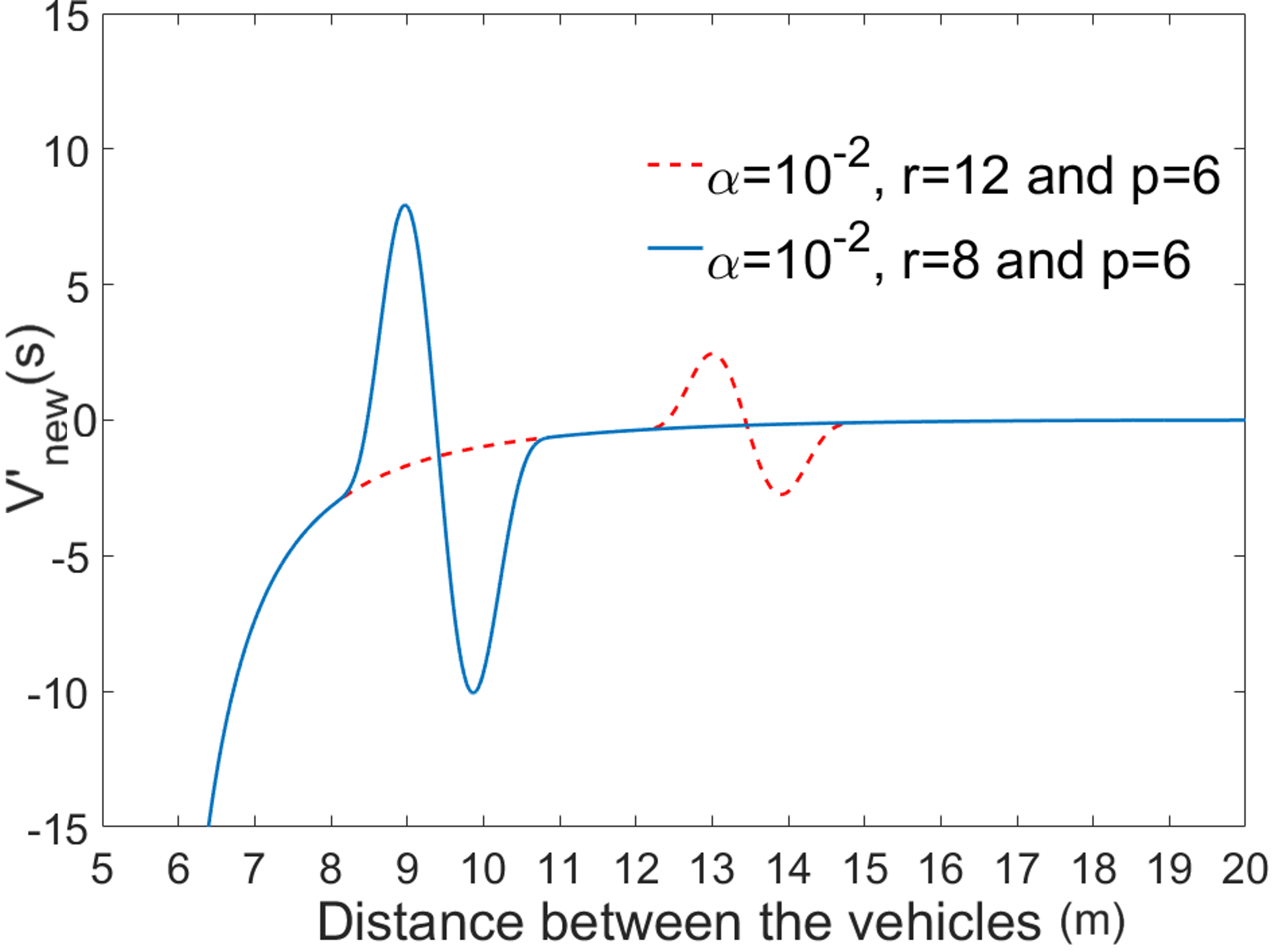}
        \caption{Influence of the parameter $r$.}
        \label{fig:Potentials_r_derivative}
    \end{subfigure}
    \hfill
    \begin{subfigure}[b]{0.32\textwidth}
        \includegraphics[width=\textwidth]{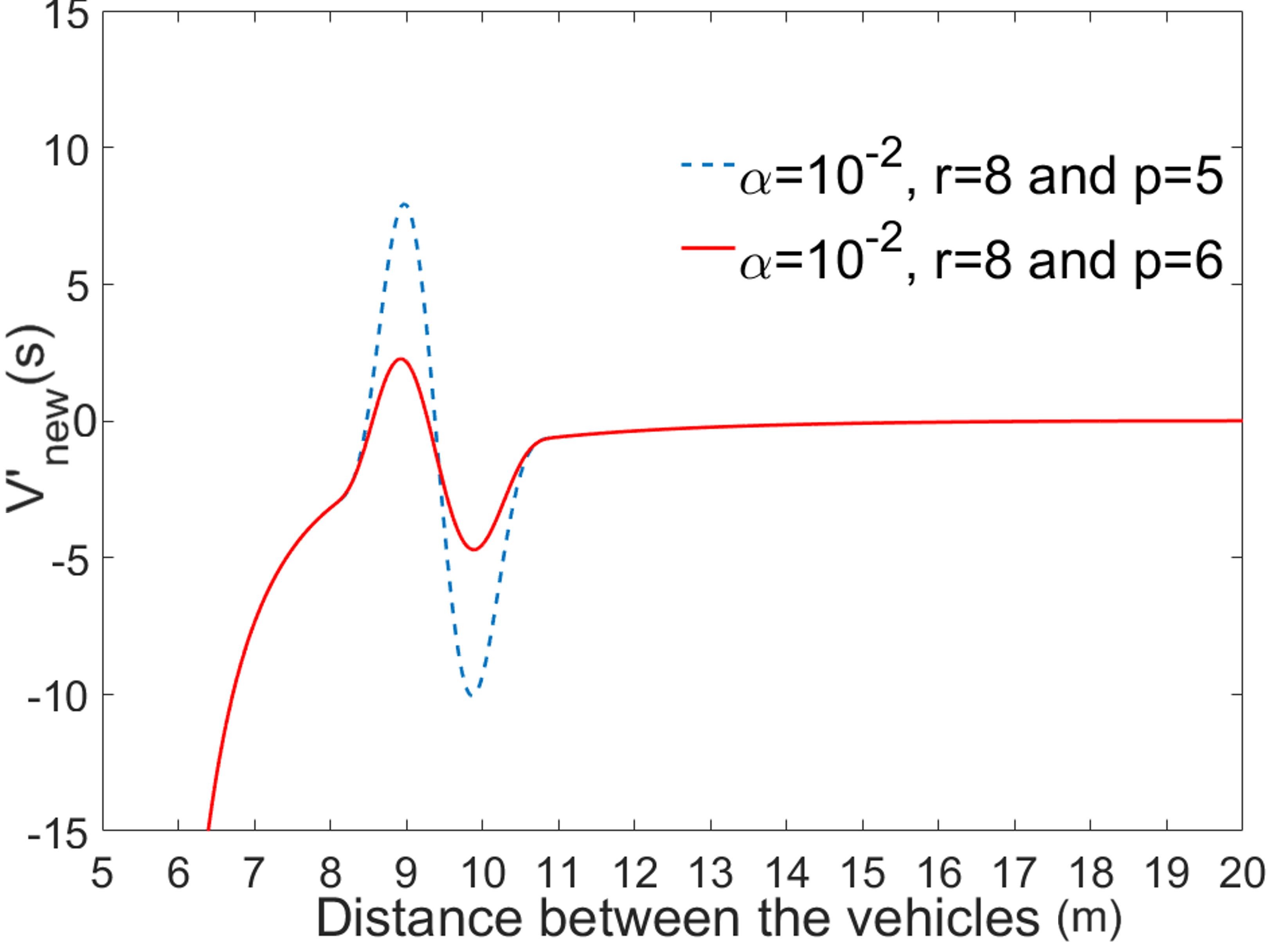}
        \caption{Influence of the parameter $p$.}
        \label{fig:Potentials_p_derivative}
    \end{subfigure}
    \caption{The derivative $V'_{new}$ of the potential $V_{new}$ under different combinations of parameters $\alpha$, $r$ and $p$.}
    \label{fig:Potentials_parameters_derivative}
\end{figure*}




\begin{remark} \label{remark_limitations}
   It was shown in \cite{karafyllis2022stability} that $V$ makes all vehicles converge to an inter-vehicle distance  {greater than or equal} to $\lambda$ at equilibrium, reducing both road capacity and overall traffic flow. Besides, the selection of $V$ may yield unrealistically large accelerations (see Section \ref{sec:sims}).
\end{remark}
Next, we propose a performance-sensitive potential function to address these issues.

\section{Performance-sensitive Potential Function}
In this section, we present the following parameterized potential function to address the limitations highlighted in Remark \ref{remark_limitations}: 
\begin{equation}
    V_{new}(s) = 
    \begin{cases} 
        \alpha \frac{  (\lambda - s)^3 }{(s - L)}, & \text{if } L < s < r,\\
        \alpha \frac{  (\lambda - s)^3 }{(s - L)} +\frac{(r+3 - s)^{p} (s - r)^{p}}{(L - s)^2}, & \text{if } r \leq s < r+3, \\
        \alpha \frac{  (\lambda - s)^3 }{(s - L)}, & \text{if } r+3 \leq s < \lambda, \\
        0, & \text{if } \lambda \geq s,
    \end{cases}
    \label{new_type_of_potential}
\end{equation}

\noindent where, recall that $s$ represents the distance between two vehicles. Then, \eqref{new_type_of_potential} gives a piece-wise potential function with varying cases corresponding to different values of the input $s$. The new potential function is parameterized by the variables $\alpha, r$ and $p\in\mathbb{R}$. The constants $\lambda$ and $L$ hold the same meaning as in \eqref{eq:V_properties}. 

Next, we examine the influence of the parameters $\alpha, r$ and $p$ on the shape of the potential function and discuss how this influence can improve the performance of our system. Hereinafter, we call the proposed potential $V$ in \eqref{new_type_of_potential} as $V_{new}$ and the potential $V$ presented in \cite{karafyllis2022stability} as $V_{old}$.

\begin{remark}
In (\ref{new_type_of_potential}), $V_{new}$ exhibits the same behavior in the first and the third case. However, in the second case, we incorporate an additional term to create a local minimum in the function followed by a local maximum, thereby forming a hill within its profile, as shown in Figs. \ref{fig:Potentials_parameters}.  
\end{remark}
\begin{remark} \label{benefit_of_V_new}
It should be noted that in \cite{karafyllis2022stability}, the potential $V_{old}$ was assumed to be decreasing, which due to \eqref{eq:initial model}, \eqref{eq:initial model1}, and \eqref{eq:V_properties} implies that the set of equilibrium points was given by $E=\{\, (s_{2} ,...,s_{n} ,v_{1} ,...,v_{n} )\in \mathbb{R} ^{2n-1} \, : v_{i} =v^{*} ,i=1,...,n,\,
 s_i\ge\lambda, i=2,\ldots,n \}$. Due to the local minimum appearing in the design of \eqref{new_type_of_potential} and due to \eqref{eq:V_properties} and \eqref{set_S}, it holds that $E\subset S$. In other words,  the potential $V_{new}$ reaches an extra equilibrium at the local minimum in comparison to $V_{old}$.
\end{remark}
Remark \ref{benefit_of_V_new} establishes that an effective parametrization of the potential $V_{new}$ can reduce inter-vehicle distances. This has a direct impact on improving road capacity and traffic flow when vehicles cruise at similar speeds.

 

\subsection{Influence of the Parameters in $V_{new}$}
In this subsection, we analyze the explicit influence of each parameter on the shape of the potential function:

\textit{1) The influence of the parameter $\alpha$} in \eqref{new_type_of_potential} is illustrated in Fig. \ref{fig:Potentials_a}. Note that changing $\alpha$ serves to scale up or down the entire span of $V_{new}$. Furthermore, the primary influence of this parameter, as shown in Fig. \ref{fig:Potentials_a_derivative}, is to either increase or decrease the slope of the curve $V'_{new}$. Thus, by selecting an appropriate value of $\alpha$, we can control the magnitude of the repulsive forces generated by $V'_{new}$ 
in \eqref{eq:initial model1}. This enables us to achieve smoother accelerations.

\textit{2) The influence of the parameter $r$} in \eqref{new_type_of_potential} is illustrated in Fig. \ref{fig:Potentials_r}.
This parameter influences the location of the local minimum along the $x-$axis by controlling the point at which $V_{new}$ switches between cases. To elucidate, consider Fig. \ref{fig:Potentials_r}. When $r=8$ (solid line), the local minimum occurs close to the distance $s=8$, while for $r=12$ (dotted line) the local minimum takes place close to $s=12$. Therefore, selecting a value of $r$ allows us to control the position of the hill induced by the second case in \eqref{new_type_of_potential}, and consequently, to control the equilibria induced by $V_{new}$.

\textit{3) The influence of the parameter $p$} is illustrated in Fig. \ref{fig:Potentials_p}. This parameter controls the magnitude of the hill's crest, that is, the local maximum induced by the second case in \eqref{new_type_of_potential}. Particularly, in Fig. \ref{fig:Potentials_p} we observe that when $p=6$ (solid line), the created hill becomes higher and sharper compared to its form when $p=5$ (dotted line). We have dynamically plotted the influence of each of these parameters on the shape of $V_{new}$, available on the supplemental website: \href{https://sites.google.com/udel.edu/idspot/home}{\underline{https://sites.google.com/udel.edu/idspot/home}}.

\begin{remark}
As a convention, we select that the range for the second case \eqref{new_type_of_potential} spans 3 m, i.e., $[r,r+3]$. This value was selected empirically to enforce a hill of a reasonable width.  While this value of 3 m can be controlled by an additional parameter, optimal adjustments are under ongoing research.
\end{remark}
\begin{remark} \label{remark_IVP}
Considering that (a) $V_{new}$ influences the feedback laws in (\ref{eq:initial model1}), (b) $V_{new}$ takes as an input the inter-vehicle distances $s$, and (c) the ODEs in (\ref{eq:initial model}) constitute an initial value problem, we can conclude that the performance of various parameter combinations varies greatly depending on the initial conditions of the vehicles. As an example, consider 2 vehicles with a distance s (input of $V_{new}$) of 8,8 m. Then, the derivative $V'_{new}$ depicted in Fig. \ref{fig:Potentials_r_derivative} would result in significantly different outcomes depending on the value of the parameter $r$ (without loss of generality, in this example, we do not consider parameters $p$ and $\alpha$). For instance, when $r=8$ (solid line), the derivative $V_{new}'$ in Fig. \ref{fig:Potentials_r_derivative}  will return a larger potential in comparison to the case where $r=12$ (dotted line). 
\end{remark}

Given Remark \eqref{remark_IVP}, we aim to identify the optimal combination of parameters for different sets of initial conditions. The significance of examining the initial conditions and associated parameters relates to the dynamic behavior of the traffic. Specifically, when a new vehicle enters a lane from an on-ramp or a different lane (for multi-lane roads), the initial value problem described in (\ref{eq:initial model}) must be resolved in order for the vehicles to update their trajectories. That is, the initial value problem is resolved, including the initial conditions associated with the new vehicle.

\subsection{Optimization Framework}

In this subsection, we formulate the following optimization problem to yield the optimal values of $\alpha$, $r$, and $p$ for given initial conditions: 
\begin{align}\label{eq:optimaztion}
\min_{\alpha,r,p} \quad  w_1&\int_{t_0}^{t_f} \sum_{i=1}^{n} \dot{v}_i^2 \,dt \; +\; w_2 \int_{t_0}^{t_f}\sum_{i=1}^{n} {s_{i}}\; dt  \\
\text{subject to:} \quad & \eqref{eq:initial model},\; s_i(t_0) = s_i^0, \; \; v_i(t_0) = v_i^0\;\; \forall i =1,\dots,n, \nonumber\\
&\max|V_{new}'(s)|\leq z \quad \forall\; s\in [r,r+3], \nonumber\\ 
&3\leq p \leq 9,\; L< r \leq \lambda-3,\; 10^{-3}\leq \alpha \leq 10^{-1}. \nonumber
\end{align}
The objective function contains two terms: the first term captures the sum of accelerations weighted by $w_1$, while the second term represents the sum of inter-vehicle distances weighted by $w_2$. Here, $t_0$ is the initial time when the problem must be solved due to a change in the traffic conditions (such as the merging of a new vehicle into the lane) and $t^f$ signifies the time horizon over which we aim to monitor the vehicles' behavior. The initial conditions for the optimization problem are denoted by the state $s_i^0$ and velocity $v_i^0$ of each vehicle $i=1,\dots,n$ at time $t_0$.


\begin{remark}
Our optimization problem (\ref{eq:optimaztion}), includes seven constraints. The first three constraints are connected to the system dynamics and the initial conditions of the vehicles while the inequality constraints define the feasibility domain for our parameters $\alpha$, $r$, and $p$.
\end{remark}
Next, we investigate the selection of the feasibility domain:

\textit{1) Slope of $V_{new}$}: We have shown that the hill in \eqref{new_type_of_potential} is responsible for driving inter-vehicle distances to one of the two equilibrium points 
while the shape of the potential $V_{new}$ decides the strength of the forces applied to the vehicles.  
 To prevent the imposition of unrealistic forces by the potential $V_{new}$ within the range of the hill, we select a threshold $z$ for the maximum derivative of the potential, i.e. $\max|V_{new}'(s)|\leq z\; \forall \; s \in [r,r+3]$. 
Note that this threshold can be changed according to the vehicles' characteristics. 

\textit{2) Parameter $p$}:
We constrain the value of $p$ 
to be greater than or equal to $3$ and smaller than or equal to $9$. For $p<3$, $V_{new}$ does not belong to $C^2$ while for $p>9$, $V'_{new}$ is always greater than 4 violating the constraint related to the slope. 

\textit{2) Parameter $r$}: Recall that $r$ is the parameter that controls the location of the hill in the range $[r,r+3]$. To ensure that the entirety of the hill stays within $\lambda$, the parameter $r$ must fall within the interval $(L, \lambda-3]$.

\textit{3) Parameter $\alpha$}: We select the domain of the scaling factor $\alpha$ to be $[10^{-3}, 10^{-1}]$. Although it would be possible to expand this range, we have chosen these specific boundaries for the following reasons: (1) scaling by a number smaller than $10^{-3}$ leads the function to be almost equal to 0, thus diminishing the function's influence on inter-vehicle distances and (2) multiplying by a number larger than $10^{-1}$ results in quite large forces, leading to rather unrealistic accelerations for inter-vehicle distances smaller than 16 m. 

\section{Solution Approach}
\subsection{Challenges in Analytical Solution}

An analytical solution to this problem includes several complexities. From (\ref{eq:initial model}), (\ref{eq:initial model1}), and (\ref{eq:initial model2}), it is evident that the feedback laws are nonlinear. Moreover, the decision variables $\alpha, r,$ and $p$ are constituents of the potential $V$, which subsequently influences the values of $v_i(t)$, $s_i(t)$, and $s_{i+1}(t)$, as shown in \eqref{eq:initial model1}. This interdependence leads to intricate interactions among the variables. For these reasons, we selected to solve the optimization problem numerically \cite{betts2005discretize}. However, solving this problem numerically in real time is computationally expensive. To circumvent this issue, we aim to solve the problem offline for many different initial conditions. Then, we generalize the results obtained from the sampled initial conditions by training a neural network to map the initial conditions to optimal parameters $\alpha$, $r$, and $p$.

\subsection{Real Time Implementation using Supervised Learning} 

To train the neural network, we generate a dataset of different initial conditions and corresponding solutions to the optimization problem. To use realistic initial conditions, we set all inter-vehicle distances $s_i$ to exceed $\overline{s}+\rho \overline{v}$, where $\overline{s}$ is a standstill distance,  $\rho$ is the minimum headway that the rear vehicle maintains following the preceding CAV and $\overline{v}$ is the initial speed of the rear CAV. 

 We generated 6,000 normalized data points and split them as follows: 85$\%$ for training, 7.5$\%$ for validation, 7.5 $\%$ for testing.  The network architecture comprised two layers with 32 and 16 neurons, respectively.  The output of each neuron was passed through a rectified linear unit (ReLU) activation function. Training of the network involved the use of backpropagation, with a mean squared error (MSE) loss function and a learning rate equal to $35 \times 10^{-6}$. The training process continued until the MSE dropped below $10^{-3}$, until it had completed 2000 epochs, or until no MSE improvement was observed for 50 consecutive epochs. Finally, the MSE achieved by the trained network was 0.0012 at epoch 1104, suggesting that, on average, each element in the predicted vectors deviates from the corresponding element in the actual vectors by $\sqrt{0.0012}$ = 0.035.

 \begin{remark}
The neural network discussed above employs a 7-vehicle case to highlight the effectiveness of our approach. However, the system can adapt to any number of vehicles on the road using a similar training approach.
 \end{remark}


\begin{figure*}[h]
    \centering
    \begin{subfigure}[b]{0.32\textwidth}
        \includegraphics[width=\textwidth]{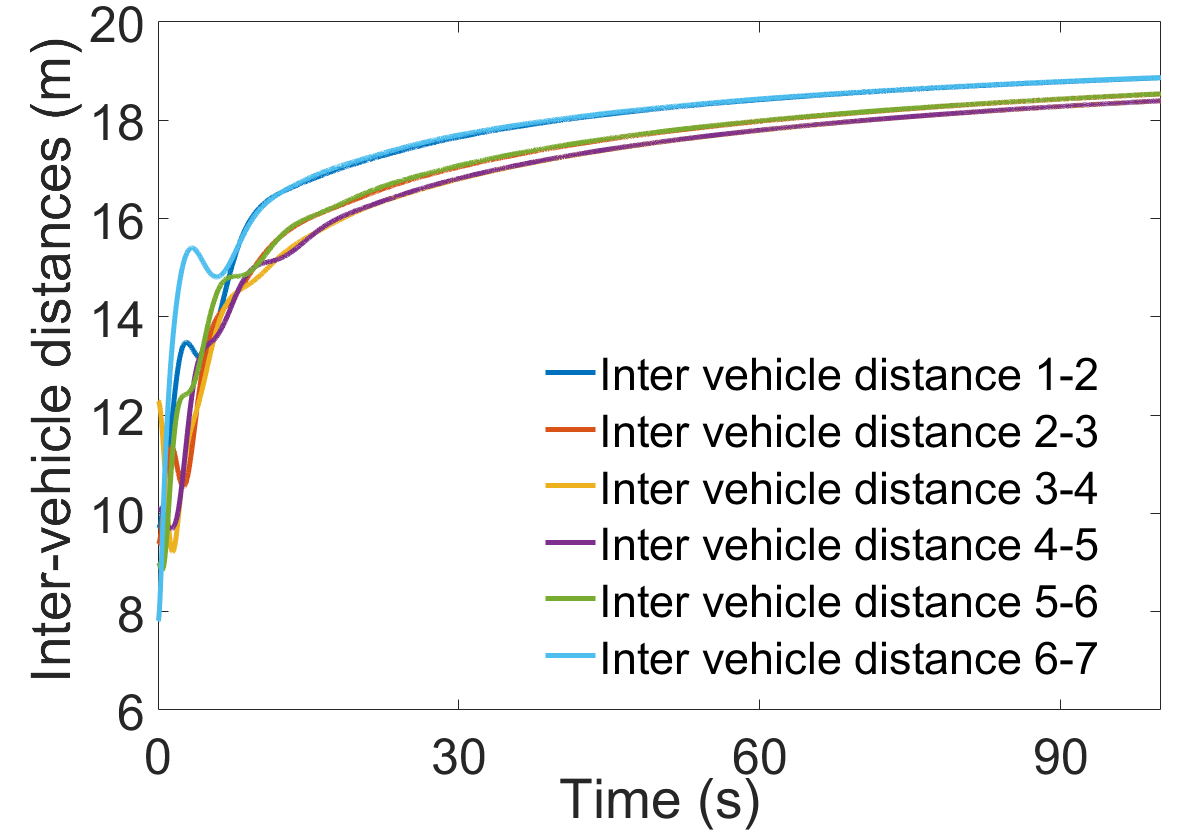}
             \caption{Inter-vehicle distances using $V_{old}$.}
        \label{fig:Inter-vehicle_old}
        
    \end{subfigure}
    \hfill
    \begin{subfigure}[b]{0.32\textwidth}
        \includegraphics[width=\textwidth]{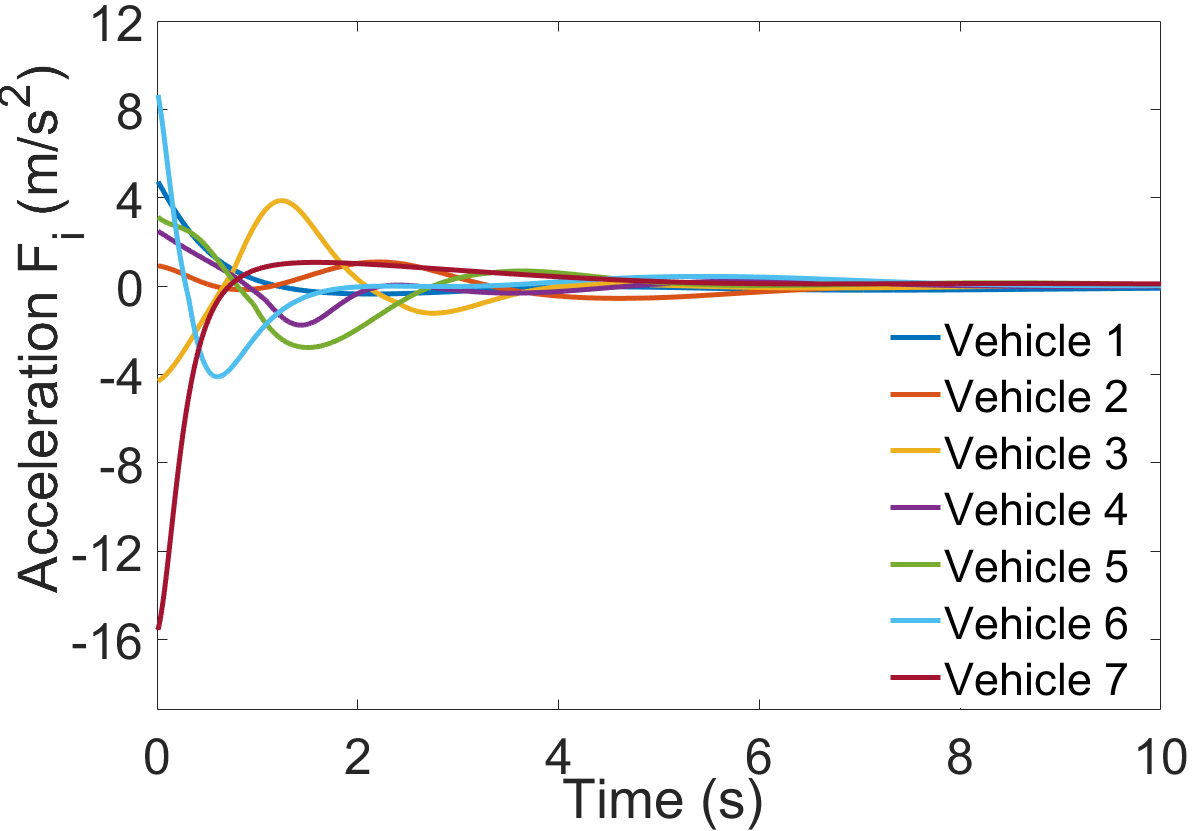}
             \caption{Accelerations with using $V_{old}$.}
        \label{fig:Acceleration_old}
        
    \end{subfigure}
    \hfill
    \begin{subfigure}[b]{0.32\textwidth}
        \includegraphics[width=\textwidth]{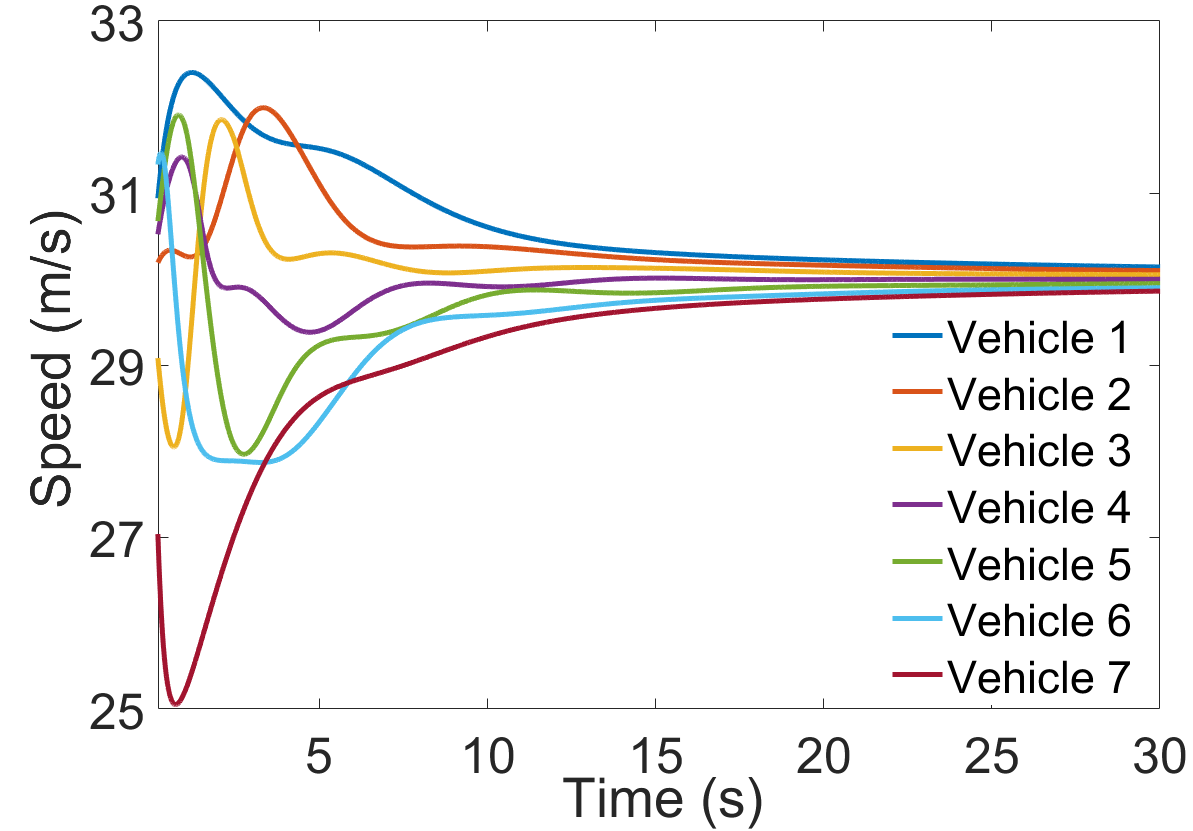}
             \caption{Speeds using $V_{old}$.}
        \label{fig:Speed old}
        
    \end{subfigure}

    \caption{Results using the potential function $V_{old}$ defined in \cite{karafyllis2022stability}.}
    \label{fig:Results_old}
\end{figure*}

\begin{figure*}[h]
    \centering
    
    \begin{subfigure}[b]{0.32\textwidth}
        \includegraphics[width=\textwidth]{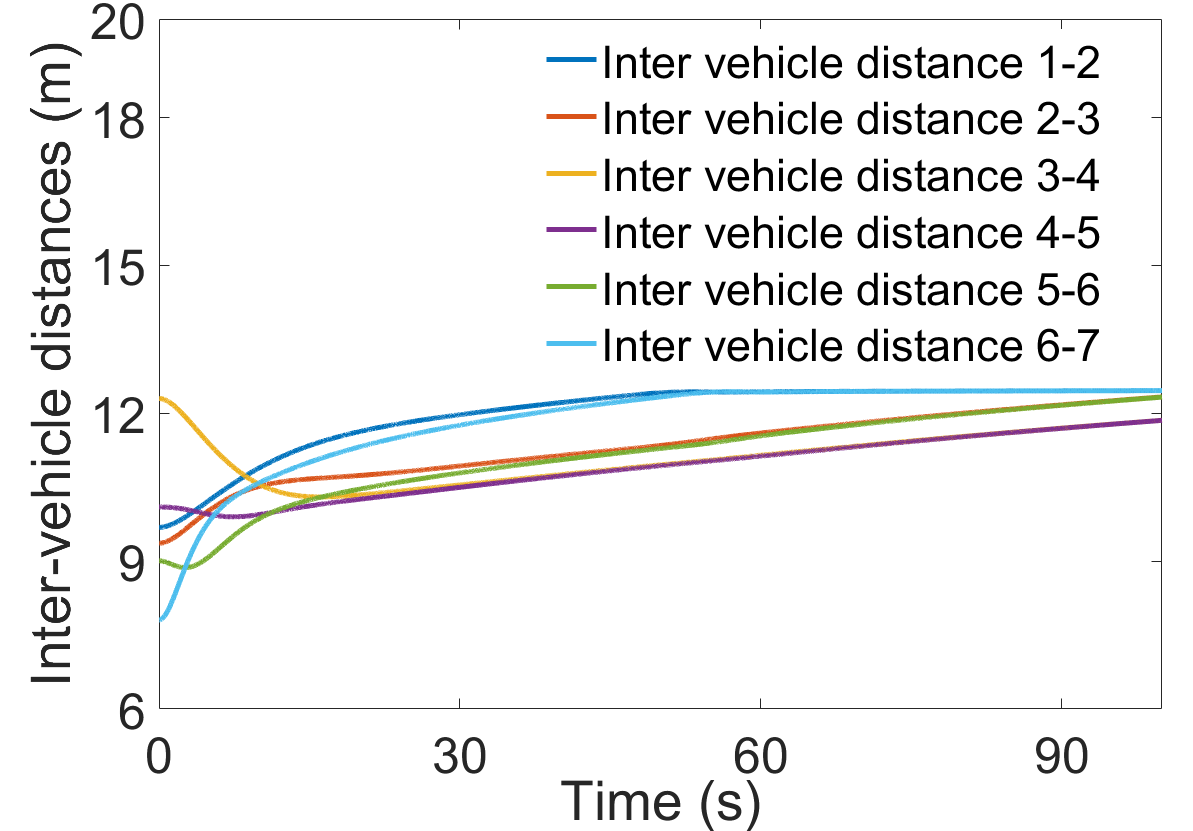}
        \caption{Inter-vehicle distances using $V_{new}$.}
        \label{fig:Inter_vehicle_new}

    \end{subfigure}
    \hspace{1mm}
    \begin{subfigure}[b]{0.32\textwidth}
        \includegraphics[width=\textwidth]{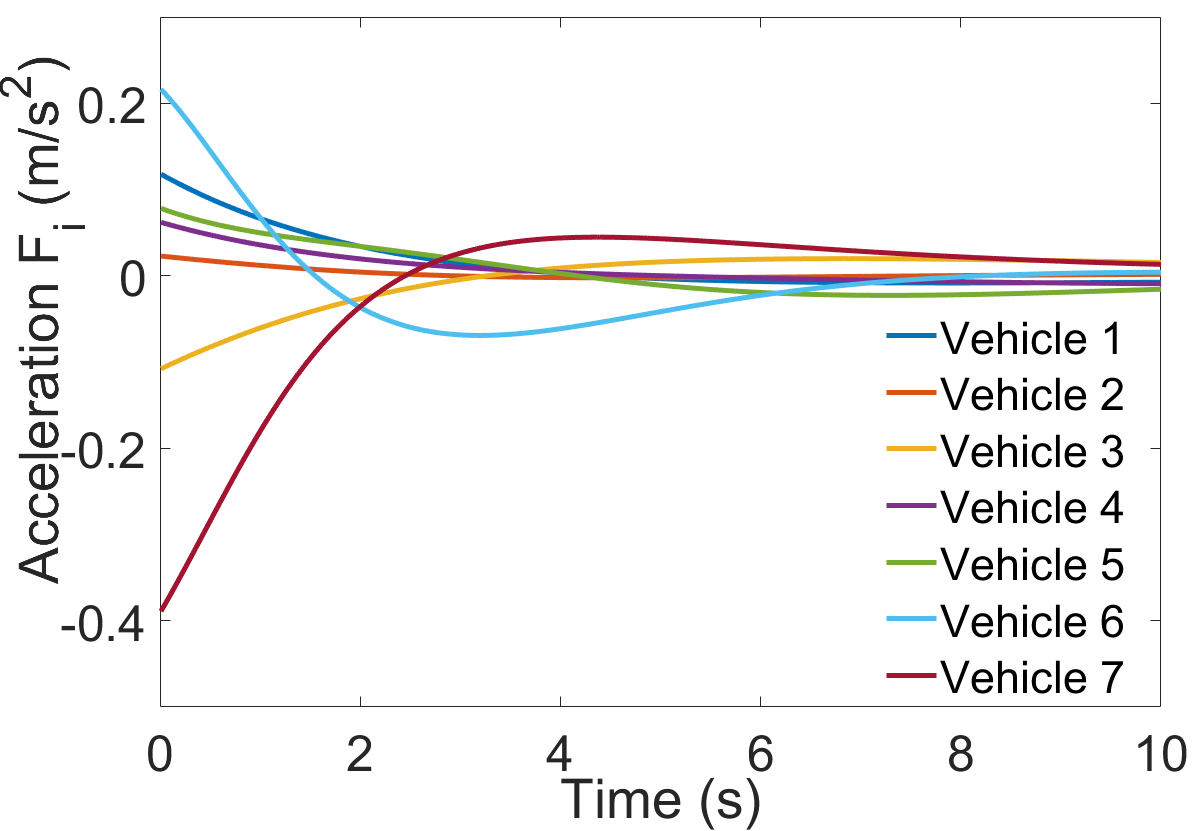}
        \caption{Accelerations using $V_{new}$.}
        \label{fig:Accelerations_new} 
    
    \end{subfigure}
    \hfill
    \begin{subfigure}[b]{0.32\textwidth}
        \includegraphics[width=\textwidth]{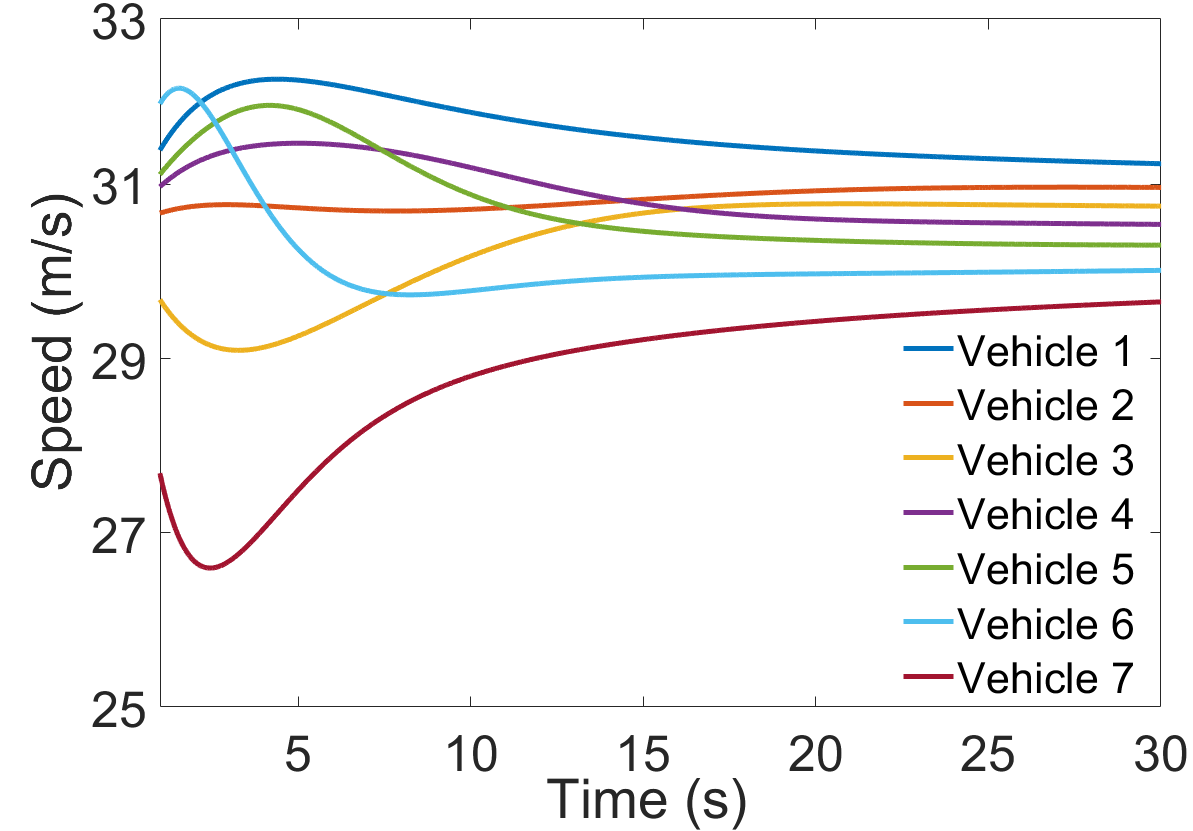}
        \caption{Speeds using $V_{new}$.}
        \label{fig:Speeds_new}     
        
    \end{subfigure}
        \caption{Results using the proposed parameterized potential function $V_{new}$.}
    \label{fig:Results_new}

\end{figure*}

\subsection{Exact Discrete Model}

To emulate our system, we need to discretize the continuous-time domain. In this subsection, we present the exact discrete model of the continuous-time model (\ref{eq:initial model}) using sampled-data feedback. Assuming that $F_i$ is constant on the interval $[t,t+T)$, where $T>0$, $t\ge0$, we obtain via direct integration of (\ref{eq:initial model}) the following exact discrete-time model:
\vspace{-5pt}
  \begin{equation} \label{Discrete_Model}
    \begin{aligned}
x_i(t+s)&=x_i(t)+Tv_i(t)+\frac{T^2}{2}F_i \\
v_i(t+s)&=v_i(t)+TF_i,
          \end{aligned}
  \end{equation}
for all $s\in[t_0,T]$.

The following lemma provides a sufficient condition for the sampling period $T$ so that a vehicle $i\in \{1,\ldots,n\}$   avoids collisions with other adjacent vehicles and has positive and bounded speed within the given speed limit $v_{\max}$. 
\begin{lemma} Let $T>0$ and consider  \eqref{eq:initial model1}, with $F_{i} (t)\equiv F_{i} $ for all $t\in [t_0,T]$,  $i=1,...,n$, and \eqref{Discrete_Model}.   Let $(s_2(t_0),\ldots,s_n(t_0),v_1(t_0),\ldots,v_n(t_0))\in \Omega $ be given, where $\Omega $ is defined in \eqref{state_space}. Let $i\in \left\{1,...,n\right\}$ be a given arbitrary index. Suppose that for each $j=1,...,n$, $j\ne i$, $v_{j} (T)\in (0,v_{\max } )$, and that the following inequalities hold 
\begin{equation} \label{T:bound:1} 
T< \frac{ \min\{ s_{i} (t_0)-L, s_{i+1}(t_0)-L\} }{v_{\max} },  
\end{equation} 
\begin{equation} \label{T:bound:2} 
-\frac{v_{i} (t_0)}{T} <F_{i} (t_0)<\frac{v_{\max } -v_{i} (t_0)}{T}.  
\end{equation} 
 Then, $v_{i} (t)\in (0,v_{\max } )$  and $s_{i} (t)>L$,  for all $t\in [t_0,T]$. 
 \end{lemma}
\begin{proof}
Let $i\in \left\{1,...,n\right\}$ and $T>0$ satisfying \eqref{T:bound:1}. Since $F_{j} (t)\equiv F_{j} $ for all $t\in [t_0,T]$, it follows by  \eqref{Discrete_Model}, the assumptions that $(s_2(t_0),\ldots,s_n(t_0),v_1(t_0),\ldots,v_n(t_0))\in \Omega $, $v_{j} (T)\in (0,v_{\max } )$  and monotonicity of $v_{j} (t)$ on $[t_0,T]$, that either  $0<v_{j} (t_0)\le v_{j} (\tau )\le v_{j} (T)<v_{\max } $, for all $\tau \in [t_0,T]$ (for non-decreasing $v_{j} (t)$, i.e., $F_{j} \ge 0$), or  $0<v_{j} (T)\le v_{j} (\tau )\le v_{j} (0)<v_{\max } $ , for all $\tau \in [t_0,T]$ (for decreasing $v_{j} (t)$, e.g., $F_{j} <0$). In any case we  conclude that $v_{j} (t)\in (0,v_{\max } )$ for all $t\in [t_0,T]$, $j\neq i$. 

Next, \eqref{Discrete_Model} and \eqref{T:bound:2} imply that $v_{i} (T)\in (0,v_{\max } )$.  Using the concluding argument above with \eqref{Discrete_Model}, where $F_{i} (t)\equiv F_{i} $, $v_{i} (t_0)\; \in (0,v_{\max } )$, and monotonicity of $v_{i} (t)$, same arguments as above show $v_{i} (t)\in (0,v_{\max } )$,  for all $t\in [t_0,T]$. 

It remains to show that $s_i (t)>L  $ for all $t\in [t_0,T]$.   Notice that $|\dot{s}_i|\leq v_{\max}$ for all $i=1,\ldots,n$.
Then, from the assumption $s_{i} (t_0)>L$, and inequality \eqref{T:bound:1} (which implies that $T<\frac{s_{i} (t_0)-L}{v_{\max} } $), we obtain $s_{i} (t)\ge s_{i} (t_0)-v_{\max}t\ge s_{i} (t_0)-v_{\max} T>L$, for all  $t\in [t_0,T]$. The proof is complete by similarly showing $s_{i+1}(t)>L$ for all $t\in [t_0,T]$.
\end{proof}

\section{Simulations} \label{sec:sims}

This section demonstrates simulations employed to authenticate the effectiveness of our proposed analysis. We consider $7$ vehicles. The parameters of our system are set to $L=5$ m, $\lambda=20$ m, $v^*=30$ m/s, $v_{max}=35$ m/s, $\epsilon=0.2$, and $\mu=0.5$ as defined in \cite{karafyllis2022stability}. We set the weights $w_1=w_2=0.5$ and $z=4$ after normalizing each term $\dot{v}_i$ and $s_{i}$ in \eqref{eq:optimaztion}. For each scenario, identical initial conditions are applied with initial inter-vehicle distances $s_i(t_0) \in (8,12)$ and initial speeds $v_i(t_0) \in (27,33)$. 

\textbf{Scenario 1:}
This scenario utilizes the potential function $V_{old}$ as defined in \cite{karafyllis2022stability}. Observing Fig. \ref{fig:Inter-vehicle_old}, we validate that the inter-vehicle distances converge towards the value $\lambda=20$ m. This is reasonable, as $V_{old}$ is strictly decreasing. Consequently, all inter-vehicle distances converge at the global minimum $\lambda$ of $V_{old}$ (see Remark 3). In this context, in Fig. \ref{fig:Acceleration_old}, we observe that the vehicle accelerations adopt large unrealistic values. This occurrence is attributed to the strong repulsive forces assumed by $V_{old}$ for short initial inter-vehicle distances. Specifically, $V_{old}$ lacked a scaling factor, leading to extremely large accelerations for inter-vehicle distances less than 16 m.

\textbf{Scenario 2:}
In this scenario, we use the proposed performance-sensitive potential function $V_{new}$. As illustrated in Fig. \ref{fig:Inter_vehicle_new}, the inter-vehicle distances converge close to 12 m. This is attributed to the introduced hill established by  $V_{new}$. Especially, per \eqref{set_S}, $V_{new}$ obtains an extra equilibrium point at its local minimum.  As a result, the inter-vehicle distances can converge to 
this local minimum. This improvement significantly increases the road capacity in comparison with Scenario 1. 
Moreover, the parameter $\alpha$ in the potential function serves as a scaling factor, facilitating the vehicles to adopt lower, more realistic acceleration rates as shown in \ref{fig:Accelerations_new}. Consequently, in Fig. \ref{fig:Speeds_new}, we see a significant advancement in the uniformity and harmonization of vehicle speeds compared to \ref{fig:Speed old} in Scenario $1$. Video simulations are available on the supplemental website: \href{https://sites.google.com/udel.edu/idspot/home}{\underline{https://sites.google.com/udel.edu/idspot/home}}.

\section{Concluding remarks and discussion}
In this paper, we built upon prior work \cite{karafyllis2022stability} and 
introduced performance-sensitive potential functions to improve vehicle accelerations and inter-vehicle distances for CAVs. Through training a neural network on an extensive offline dataset, we bypassed the computational burden of driving the optimal solution in real time, yielding smoother accelerations, improved traffic flow, and enhanced speed harmonization. Furthermore, we established conditions for the sampled data model that guarantee collision avoidance, emphasizing the model's safety considerations. Future work will extend this framework to lane-free roads and consider the integration of human-driven vehicles.

\linespread{0.99}\selectfont
\bibliographystyle{ieeetr}
\bibliography{bibliography, IDS_Publications_08282023}

\begin{thebibliography}{10}

\bibitem{othman2019ecological}
B.~Othman, G.~De~Nunzio, D.~Di~Domenico, and C.~Canudas-de Wit, ``Ecological
  traffic management: A review of the modeling and control strategies for
  improving environmental sustainability of road transportation,'' {\em Annual
  Reviews in Control}, vol.~48, pp.~292--311, 2019.

\bibitem{rios2016survey}
J.~Rios-Torres and A.~A. Malikopoulos, ``A survey on the coordination of
  connected and automated vehicles at intersections and merging at highway
  on-ramps,'' {\em IEEE Transactions on Intelligent Transportation Systems},
  vol.~18, no.~5, pp.~1066--1077, 2016.

\bibitem{Malikopoulos2020}
A.~A. Malikopoulos, L.~E. Beaver, and I.~V. Chremos, ``Optimal time trajectory
  and coordination for connected and automated vehicles,'' {\em Automatica},
  vol.~125, no.~109469, 2021.

\bibitem{malikopoulos2018decentralized}
A.~A. Malikopoulos, C.~G. Cassandras, and Y.~J. Zhang, ``A decentralized
  energy-optimal control framework for connected automated vehicles at
  signal-free intersections,'' {\em Automatica}, vol.~93, pp.~244--256, 2018.

\bibitem{xiao2021bridging}
W.~Xiao, C.~G. Cassandras, and C.~A. Belta, ``Bridging the gap between optimal
  trajectory planning and safety-critical control with applications to
  autonomous vehicles,'' {\em Automatica}, vol.~129, p.~109592, 2021.

\bibitem{xiao2010comprehensive}
L.~Xiao and F.~Gao, ``A comprehensive review of the development of adaptive
  cruise control systems,'' {\em Vehicle system dynamics}, vol.~48, no.~10,
  pp.~1167--1192, 2010.

\bibitem{ioannou1993autonomous}
P.~A. Ioannou and C.-C. Chien, ``Autonomous intelligent cruise control,'' {\em
  IEEE Transactions on Veh. technology}, vol.~42, no.~4, pp.~657--672, 1993.

\bibitem{van2006impact}
B.~Van~Arem, C.~J. Van~Driel, and R.~Visser, ``The impact of cooperative
  adaptive cruise control on traffic-flow characteristics,'' {\em IEEE
  Transactions on Intel. Transportation Systems}, vol.~7, no.~4, pp.~429--436,
  2006.

\bibitem{bae2019design}
S.~Bae, Y.~Kim, J.~Guanetti, F.~Borrelli, and S.~Moura, ``Design and
  implementation of ecological adaptive cruise control for autonomous driving
  with communication to traffic lights,'' in {\em 2019 American Control
  Conference (ACC)}, pp.~4628--4634, IEEE, 2019.

\bibitem{zhang2020cooperative}
Y.~Zhang, Y.~Bai, M.~Wang, and J.~Hu, ``Cooperative adaptive cruise control
  with robustness against communication delay: An approach in the space
  domain,'' {\em IEEE Transactions on Intelligent Transportation Systems},
  vol.~22, no.~9, pp.~5496--5507, 2020.

\bibitem{zhai2018ecological}
C.~Zhai, F.~Luo, Y.~Liu, and Z.~Chen, ``Ecological cooperative look-ahead
  control for automated vehicles travelling on freeways with varying slopes,''
  {\em IEEE Transactions on Vehicular Technology}, vol.~68, no.~2,
  pp.~1208--1221, 2018.

\bibitem{vahidi2018energy}
A.~Vahidi and A.~Sciarretta, ``Energy saving potentials of connected and
  automated vehicles,'' {\em Transportation Research Part C: Emerging
  Technologies}, vol.~95, pp.~822--843, 2018.

\bibitem{mahbub2020sae-1}
A.~M.~I. Mahbub, V.~Karri, D.~Parikh, S.~Jade, and A.~A. Malikopoulos, ``A
  decentralized time- and energy-optimal control framework for connected
  automated vehicles: From simulation to field test,'' in {\em SAE Technical
  Paper 2020-01-0579}, SAE International, 2020.

\bibitem{karafyllis2022stability}
I.~Karafyllis, D.~Theodosis, and M.~Papageorgiou, ``Stability analysis of
  nonlinear inviscid microscopic and macroscopic traffic flow models of
  bidirectional cruise-controlled vehicles,'' {\em IMA Journal of Mathematical
  Control and Information}, vol.~39, no.~2, pp.~609--642, 2022.

\bibitem{theodosis2022sampled}
D.~Theodosis, F.~N. Tzortzoglou, I.~Karafyllis, I.~Papamichail, and
  M.~Papageorgiou, ``Sampled-data controllers for autonomous vehicles on
  lane-free roads,'' in {\em 2022 30th Mediterranean Conference on Control and
  Automation (MED)}, pp.~103--108, IEEE, 2022.

\bibitem{tzortzoglou2023approach}
F.~N. Tzortzoglou, D.~Theodosis, and A.~Malikopoulos, ``An approach for
  optimizing acceleration in connected and automated vehicles,'' {\em arXiv
  preprint arXiv:2308.04632}, 2020 (in review).

\bibitem{karafyllis2022lyapunov}
I.~Karafyllis, D.~Theodosis, and M.~Papageorgiou, ``Lyapunov-based
  two-dimensional cruise control of autonomous vehicles on lane-free roads,''
  {\em Automatica}, vol.~145, p.~110517, 2022.

\bibitem{betts2005discretize}
J.~T. Betts and S.~L. Campbell, ``Discretize then optimize,'' {\em Mathematics
  for industry: challenges and frontiers}, pp.~140--157, 2005.

\end{thebibliography}

\end{document}